%% file: SSFM.tex
\documentclass[11pt]{article}
\input{_config}

\def\R{\mathbb{R}}
\def\OptSet{{\sf OptSet}}
\def\adj{\mathrm{adj}}
\usepackage{algorithm,algpseudocode,amssymb,amsmath}

\title{A Unified Approach to Minimizing \\
Symmetric Submodular Functions}

\author{Satoru Iwata
\thanks{Department of Mathematical Informatics, 
The University of Tokyo, Tokyo 113-8656, Japan, 
Institute for Chemical Reaction Design and Discovery, Hokkaido University, Sapporo, Hokkaido 001-0021, Japan.
{iwata@mist.i.u-tokyo.ac.jp}} \and
Haruto Konno \thanks{
Department of Mathematical Informatics, 
The University of Tokyo, Tokyo 113-8656, Japan. 
{fruitmix246@g.ecc.u-tokyo.ac.jp}}}

\begin{document}
\maketitle
\begin{abstract}
Symmetric submodular function minimization admits purely combinatorial algorithms using special orderings of the ground set.
Extending the minimum-cut algorithm of Nagamochi and Ibaraki (1992), Queyranne (1998) showed that the maximum adjacency ordering yields a pendent pair, which can be used to find a nontrivial minimizer.
Nagamochi (2010) later introduced the minimum degree ordering, which yields a flat pair and leads to the identification of extreme sets.
Despite the apparent similarity between these two algorithms, their connection remained unclear.  
  
In this paper, we introduce yet another ordering called minimum capacity ordering, and extend it to a one-parameter family of orderings, called $\alpha$-orderings, that unifies these two previously known orderings.
We prove a general inequality for $\alpha$-orderings, and 
our framework recovers the known pendent-pair and flat-pair results as special cases, corresponding to $\alpha = -1$ and $\alpha = 1$, respectively.
For each $\alpha \in [-1, 1]$, the last two elements of an $\alpha$-ordering form a contractible pair, i.e., a pair whose contraction preserves the existence of a nontrivial minimizer, which leads to a contraction algorithm that finds a nontrivial minimizer of a symmetric submodular function in $O(n^3)$ oracle calls, where $n$ is the cardinality of the ground set.
In addition, we discuss the ranges of $\alpha$ that ensure $\alpha$-ordering to obtain these special pairs. 
\end{abstract}

\section{Introduction}
Let $V$ be a finite set of cardinality $n$. We denote by $\mathbb{R}$ and $\mathbb{R}_+$ the sets of reals and nonnegative reals, respectively.
A set function $f: 2^V \to \mathbb{R}$ is called \emph{submodular} if 
\begin{equation*}
  f(X) + f(Y) \geq f(X \cup Y) + f(X \cap Y),
  \qquad \forall X, Y \subseteq V.
\end{equation*}
Submodularity is a fundamental notion in discrete optimization.
It has been studied through central structures such as
matroids, polyhedra, and convexity \cite{Edmonds1970Submodular,fujishige2005submodular,lovasz1983submodular}.
The general problem of minimizing a submodular function has a long history.
The first polynomial-time algorithm was obtained via the ellipsoid method by Grötschel, Lovász, and Schrijver~\cite{grotschel1981ellipsoid},
and combinatorial strongly polynomial algorithms were developed later by 
Iwata, Fleischer, and Fujishige~\cite{iwata2001combinatorial} and by Schrijver~\cite{schrijver2000combinatorial}, followed by subsequent improvements~\cite{dadush2021rescaling,iwata2003faster,iwataorlin2009simple,lee2015cuttingplane,orlin2009faster}. The current best bound on oracle calls in a strongly polynomial algorithm is $O(n^3\log\log n/\log n)$ due to Jiang~\cite{jiang2023convex}.

In this paper, we focus on the symmetric case.
A set function $f:2^V \to \mathbb{R}$ is called \emph{symmetric} if
\begin{equation*}
  f(X) = f(V \setminus X),
  \qquad \forall X \subseteq V.
\end{equation*}
Symmetric submodular functions form an important subclass of submodular functions.
Fujishige~\cite{fujishige1983canonical} studied their structural properties and established a decomposition theory for them. 
For a symmetric submodular function $f$, both $\varnothing$ and $V$ are always minimizers. Accordingly, the symmetric submodular function minimization problem asks for a nontrivial minimizer, i.e., a proper nonempty subset that achieves the minimum value. 

A canonical example of a symmetric submodular function is the cut capacity function of a capacitated undirected graph. For an undirected graph $G = (V, E)$ with an edge capacity function $c : E \to \mathbb{R}_+$, its cut capacity function
$\kappa: 2^V \to \mathbb{R}_+$ is defined by
\begin{equation*}
  \kappa(X) := \sum_{\{u, v\} \in E:\, |\{u, v\} \cap X| = 1} c(u, v)
  \qquad (\forall X \subseteq V),
\end{equation*}
which is symmetric and submodular.
The minimization of the cut function over nonempty proper subsets coincides with the global minimum-cut problem in an undirected graph.
A straightforward polynomial-time approach to this problem is to fix a vertex $s \in V$,
compute a minimum $(s,t)$-cut for each $t \in V \setminus \{s\}$ by a max-flow algorithm, and take the best one.

A key combinatorial approach to this problem originates in Nagamochi and Ibaraki~\cite{NagamochiIbaraki1992EdgeConnectivity},
whose correctness proof was later simplified by Frank~\cite{frank94edge-connectivity} and by Stoer and Wagner~\cite{StoerWagner1997SimpleMinCut}.
The method is based on constructing a \emph{maximum adjacency ordering} (MA-ordering) of the vertices and
using the fact that the last two vertices in such an ordering form a \emph{pendent pair}.
This implies that, among all cuts separating these two vertices, a minimum one is already represented by a singleton side,
so the pair can be contracted without losing any information needed to recover a global minimum cut.
Repeating this contraction step leads to one of the simplest and fastest deterministic global minimum-cut algorithms.

Queyranne~\cite{Queyranne1998} extended this ordering-based contraction framework from graph cut functions to arbitrary symmetric submodular functions,
thereby obtaining a combinatorial strongly polynomial algorithm that finds a nontrivial minimizer in $O(n^3)$ oracle calls.
For symmetric submodular function minimization, this approach remains substantially simpler and faster than using strongly polynomial algorithms for general submodular function minimization $O(n)$ times. 

More recently, Nagamochi~\cite{Nagamochi2010MinimumDegreeOrderings} introduced the \emph{minimum degree ordering} (MD-ordering) and proved that the last two vertices in such an ordering form a \emph{flat pair}.
In the graph setting, contracting flat pairs yields not only a minimum cut but also all \emph{extreme sets}
(extreme subsets, or LS-sets), which capture a finer structure of the connectivity function.
The notion of extreme sets predates this ordering-based framework:
it was introduced by Luccio and Sami~\cite{luccio1969decomposition}
as a compact representation of graph connectivity through decomposition into minimally interconnected subnetworks,
and was later applied by Lawler~\cite{lawler1973cutsets} to hypergraph partitioning.
Nagamochi~\cite{Nagamochi2010MinimumDegreeOrderings} further extended this viewpoint to symmetric submodular functions,
showing that the ordering-based contraction framework can recover not only a nontrivial minimizer
but also the finer structural information encoded by extreme sets with $O(n^3)$ oracle calls.

Exploiting the property of pendent pairs obtained by MA-ordering, Goemans and Soto~\cite{goemans2013hereditary} presented an algorithm that finds a  nontrivial minimizer of a symmetric submodular function in a hereditary family of subsets using $O(n^3)$ oracle calls. They further extended this algorithm to enumerate all minimal such optimal sets in the same asymptotic running time. Unlike symmetric submodular function minimization, the problem itself breaks symmetry, and the MD-ordering has not been applied to solve the same problem. In contrast, Goemans and Soto~\cite{goemans2013hereditary} also described an algorithm that enumerates all extreme sets of a symmetric submodular function using $O(n^4)$ oracle calls, whereas Nagamochi's algorithm achieves the same goal more efficiently. Thus, MA- and MD-orderings have their own advantages. 

From a higher perspective, however, the similarity between the algorithms of Queyranne and Nagamochi suggests that the MA- and MD-orderings should be viewed not as isolated constructions, but as manifestations of a broader common principle. The purpose of this paper is to clarify their common structure. 

To capture the algorithmic content of the pairs of elements obtained at the end of these orderings, we introduce the notion of a \emph{contractible pair}.
This is a weaker property than being a pendent pair or a flat pair, but it is still sufficient to support a contraction algorithm for finding a nontrivial minimizer. We also introduce yet another ordering called the \emph{minimum capacity ordering} (MC-ordering). Our first main result (Theorem~\ref{th:MC-ordering}) shows that the last two elements of an MC-ordering form a contractible pair, thereby obtaining the third ordering-based combinatorial algorithm for finding a nontrivial minimizer of a symmetric submodular function with $O(n^3)$ oracle calls. 

We then extend the MC-ordering to a one-parameter family of orderings, called \emph{$\alpha$-orderings}, which interpolates between the two previously known orderings: the case $\alpha=-1$ coincides with the MA-ordering, 
while the case $\alpha=1$ corresponds to the MD-ordering. 
Our second main result is a general inequality for $\alpha$-orderings (Theorem~\ref{thm:alpha-order-cut}). From this inequality, the pendent-pair theorem of Queyranne and the flat-pair theorem of Nagamochi are both recovered as special cases. In addition, for each $\alpha \in [-1,1]$, the last two elements of an $\alpha$-ordering form a contractible pair (Corollary~\ref{cor:find-contractible}). As a consequence, repeated contraction yields a nontrivial minimizer of a symmetric submodular function with $O(n^3)$ oracle calls. 

We also present complementary converse results showing that this range of parameters is tight in a natural sense: for $\alpha \notin [-1,1]$, the last two elements of an $\alpha$-ordering need not be contractible. 
Moreover, for $\alpha \neq -1$ and $\alpha \neq 1$, they need not be a pendent pair or a flat pair, respectively. 

The outline of this paper is as follows. Section \ref{sec:prel} summarizes the two previously known algorithms for symmetric submodular function minimization. In Section \ref{sec:CP}, we introduce the concept of contractible pairs, and then present how to use those pairs for minimizing symmetric submodular functions. In Section~\ref{sec:MC-ordering}, we introduce an MC-ordering, and show that the last two elements form a contractible pair. 
Section~\ref{sec:alpha-ordering} extends MC-ordering to $\alpha$-ordering that unifies the MA- and MD-ordering as its special cases, and shows that $\alpha$-ordering for any $\alpha\in[-1,1]$ yields a contractible pair. Finally, in Section~\ref{sec:cond}, we examine the ranges of $\alpha$ that ensure $\alpha$-ordering to obtain contractible, pendent, and flat pairs.  

\section{Maximum Adjacency Ordering and Minimum Degree Ordering}
\label{sec:prel}
In this section, we briefly summarize the two previously known ordering-based algorithms for symmetric submodular function minimization. 
Given an ordering $(v_1,\ldots,v_n)$ of $V$, we denote $V_0 := \varnothing$ and $V_i := \{v_1, \dots, v_i\}$ for $i = 1, \dots, n$.

Consider an undirected graph $G=(V,E)$ with an edge capacity function $c:E\to\R_+$. For a subset $W\subsetneq V$ and a vertex $y\in V\setminus W$, let $\adj(W,y)$ denote the adjacency between $W$ and $y$, i.e., the sum of $c(x,y)$ for all edges $\{x,y\}\in E$ with $x\in W$. Note that $\adj(W,y)$ can be expressed in terms of the cut capacity function $\kappa$ by 
$$\adj(W,y)=\frac{1}{2}\left[\kappa(\{y\})+\kappa(W)-\kappa(W\cup\{y\})\right].$$
Starting with an arbitrary vertex $v_1\in V$, construct an ordering $(v_1,\ldots,v_n)$ of $V$ in such a way that $v_i$ attains the maximum value of $\adj(V_{i-1},y)$ among all $y\in V\setminus V_{i-1}$ for $i=2,\ldots,n$. The key observation of Nagamochi and Ibaraki~\cite{NagamochiIbaraki1992EdgeConnectivity} is that the last two vertices $v_{n-1}$ and $v_n$ of this maximum adjacency ordering form a pendent pair, i.e., $\kappa(X)\geq \kappa(\{v_n\})$ holds for any $X\subsetneq V$ with $|X\cap\{v_{n-1},v_n\}|=1$. This allows us to contract the last two vertices into one without losing any candidates of the minimum-cut value other than $\kappa(\{v_n\})$. 

Stimulated by simple correctness proofs of the Nagamochi--Ibaraki algorithm given by Frank~\cite{frank94edge-connectivity} and by Stoer and Wagner~\cite{StoerWagner1997SimpleMinCut}, Queyranne~\cite{Queyranne1998} extended the concepts of maximum adjacency ordering and pendent pair to the general setting of symmetric submodular functions, and presented the first combinatorial algorithm for minimizing them. 

\begin{definition}[Maximum Adjacency Ordering]
  An ordering $(v_1, \dots, v_n)$ of $V$ is called a maximum adjacency ordering (MA-ordering) with respect to a set function $f$ on $V$ if
  \begin{equation*}\label{eq:ma-order}
    f(\{v_i\}) - f(V_{i - 1} \cup \{v_i\}) 
    \geq f(\{v_j\}) - f(V_{i - 1} \cup \{v_j\})
  \end{equation*}
  holds for every pair $(i,j)$ such that $1 \leq i \leq j \leq n$. 
\end{definition}

\begin{definition}[Pendent Pair]
  A pair $(u, v)$ of distinct elements of $V$ is called a pendent pair of a set function $f$ on $V$ if 
  \begin{equation*}\label{eq:pendent}
    f(X) \geq f(\{v\})
  \end{equation*}
  holds for every subset $X \subsetneq V$ with $|X \cap \{u, v\}| = 1$. 
\end{definition}

\begin{theorem}[Queyranne~\cite{Queyranne1998}]\label{thm:ma-order-cut}
  Let $f: 2^V \to \mathbb{R}$ be a symmetric submodular function
  and $(v_1, \dots, v_n)$ be an MA-ordering of $V$ with respect to $f$.
  Then $(v_{n - 1}, v_n)$ is a pendent pair.
\end{theorem}

Somewhat later, Nagamochi~\cite{Nagamochi2010MinimumDegreeOrderings}
designed another ordering-based algorithm for finding all the extreme sets in a capacitated undirected graph. For each vertex $y\in V$, the degree of $y$ means the sum of all the capacities of edges incident to $y$. The algorithm repeatedly selects a vertex that attains the minimum degree and removes it from the graph. Let $\deg_W(y)$ denote the degree of $y$ in the graph obtained by removing $W$ from $G$. Note that $\deg_W(y)$ can be expressed in terms of $\kappa$ by 
$$\deg_W(y)=\frac{1}{2}[\kappa(\{y\})-\kappa(W)+\kappa(W\cup\{y\})].$$
The two vertices $u$ and $v$ remained at the end of this procedure are shown to form a flat pair, i.e., any subset $X\subsetneq V$ with $|X\cap \{u,v\}|=1$ satisfies $\kappa(X)\geq\min\{\kappa(\{x\})\mid x\in X\}$. Then the algorithm contracts $u$ and $v$ into one. Repeating this process $O(n)$ times enables us to enumerate all the extreme sets. Nagamochi~\cite{Nagamochi2010MinimumDegreeOrderings} also extended this algorithm to the setting of symmetric submodular functions as follows. 

\begin{definition}[Minimum Degree Ordering]
  An ordering $(v_1, \dots, v_n)$ of $V$ is called a minimum degree ordering (MD-ordering) with respect to a set function $f$ on $V$ if
  \begin{equation*}\label{eq:md-order}
    f(V_{i - 1} \cup \{v_i\}) + f(\{v_i\}) 
    \leq f(V_{i - 1} \cup \{v_j\}) + f(\{v_j\})
  \end{equation*}
  holds for every pair $(i,j)$ such that $1 \leq i \leq j \leq n$.
\end{definition}

\begin{definition}[Flat Pair]
  A pair $(u, v)$ of distinct elements of $V$ is called a flat pair with respect to a set function $f$ on $V$ if 
  \begin{equation*}\label{eq:flat}
    f(X) \geq \min_{x \in X} f(\{x\})
  \end{equation*}
  holds for every subset $X \subsetneq V$ with $|X \cap \{u, v\}| = 1$.
\end{definition}

\begin{theorem}[Nagamochi~\cite{Nagamochi2010MinimumDegreeOrderings}]\label{thm:md-order-cut}
  Let $f: 2^V \to \mathbb{R}$ be a symmetric submodular function
  and $(v_1, \dots, v_n)$ be an MD-ordering of $V$ with respect to $f$.
  Then $(v_{n - 1}, v_n)$ is a flat pair.
\end{theorem}

Besides Queyranne's original argument, alternative expositions and proofs for \autoref{thm:ma-order-cut} are known in the literature \cite{fujishige1998another,rizzi2000symmetric}.
Our approach below is different in nature: we do not treat the pendent-pair and flat-pair results separately. Instead, we derive both of them from a single inequality for a parametrized family of orderings.

\section{Contractible Pair}
\label{sec:CP}
In this section, we introduce the concept of contractible pairs for set functions $f$ on $V$, and we present an algorithm that finds a nontrivial minimizer of a symmetric submodular function $f$ by iteratively contracting such pairs. 

\begin{definition}[Contractible Pair]
  A pair $(u,v)$ of distinct elements in $V$ is called a contractible pair with respect to a set function $f$ on $V$ if 
  \begin{equation*}\label{eq:contractible}
    f(X) \geq \min_{x \in V} f(\{x\})
  \end{equation*}  
  holds for every subset $X \subsetneq V$ with $|X \cap \{u, v\}| = 1$. 
\end{definition}

We present a recursive algorithm using contractible pairs to find a nontrivial minimizer of $f$. It first finds a contractible pair $(u,v)$. 
Let $f'$ be the set function on $V':=V\setminus\{u,v\}\cup\{w\}$ defined by 
\begin{equation*}
  f'(X) :=
  \begin{cases}
    f(X)             & (w \notin X \subseteq V'), \\
    f(V'\setminus X) & (w \in X \subseteq V').
  \end{cases}
\end{equation*}
Obviously, $f'$ is a symmetric submodular function on $V'$. Let $Y\subsetneq V'\setminus\{w\}$ be a nontrivial minimizer of $f'$. Since $(u,v)$ is a contractible pair, $\{u\}$, $\{v\}$, or $Y$ attains the minimum value of $f$ among all proper nonempty sets of $V$. This suggests a recursive algorithm for symmetric submodular function minimization, formally described as follows.  
\begin{algorithm}[H]
\caption{$\OptSet(V,f)$}
\label{alg:OptSet}
\begin{algorithmic}[1]
      \State Find a contractible pair $(u,v)$ of $f$.       
      \If{$V=\{u,v\}$} $Y\gets \{u\}$ 
      \Else 
      \State $V' \gets (V \setminus \{u, v\}\bigr) \cup \{w\}$ 
      \State Let $f': 2^{V'} \to \mathbb{R}$ be the function  obtained from $f$ by contracting $u$ and $v$ into $w$. 
      \State $Y\gets \OptSet(V',f')$ 
      \If{$w\in Y$} $Y\gets V'\setminus Y$ \EndIf 
      \If{$f(\{u\})<f(Y)$} $Y\gets\{u\}$ \EndIf 
      \If{$f(\{v\})<f(Y)$} $Y\gets\{v\}$ \EndIf 
    \EndIf 
    \State \Return $Y$.
\end{algorithmic}
\end{algorithm}

\begin{theorem}\label{thm:optimal-set-contractibility}
  Let $f: 2^V \to \R$ be a symmetric submodular function.
  Then the algorithm $\OptSet(V,f)$ returns a nontrivial minimizer of $f$ with $O(n^3)$ oracle calls.
\end{theorem}

\section{Minimum Capacity Ordering}
\label{sec:MC-ordering}
In this section, we introduce yet another ordering called minimum capacity ordering (MC-ordering), which finds a contractible pair. 

\begin{definition}[Minimum Capacity Ordering]
  An ordering $(v_1, \dots, v_n)$ of $V$ is called a minimum capacity ordering (MC-ordering) with respect to a set function $f:2^V\to\R$ if
  \begin{equation*}\label{def:alpha-order}
    f(V_{i - 1} \cup \{v_i\}) 
    \leq f(V_{i - 1} \cup \{v_j\})
  \end{equation*}
  holds for every pair $(i,j)$ such that $1 \leq i \leq j \leq n$. 
\end{definition}

Given $V_{i - 1}$, one can choose the next element $v_i$ from $V \setminus V_{i - 1}$ by evaluating $f(V_{i - 1} \cup \{u\})$ for all $u \in V \setminus V_{i - 1}$. Thus, an MC-ordering can be obtained with $O(n^2)$ oracle calls.

\begin{lemma}\label{lem:MC-ordering}
  Let $f: 2^V \to \mathbb{R}$ be a submodular set function on a finite set $V$ of cardinality $n$ and $(v_1, \dots, v_n)$ be an MC-ordering of $V$. 
  For an  arbitrary subset $X \subsetneq V$ such that $|X \cap \{v_{n - 1}, v_n\}| = 1$, we have 
  \begin{equation}\label{eq:cut-0-order}
    f(X) + f(V \setminus X) \geq \min_{x \in X} f(\{x\}) + f(V_{n - 1}).
  \end{equation}
\end{lemma}
\begin{proof}
We intend to show the statement by induction on $n=|V|$. If $n=2$, the statement obviously holds. For $n>2$, consider a submodular set function $f':2^{V'}\to\R$ on $V':=V\setminus\{v_1\}$ defined by 
$$f'(Y):=f(Y\cup\{v_1\})-f(\{v_1\})\quad\quad(\forall Y\subseteq V').$$
If $v_1\notin X$, by the inductive assumption, we have  
$$f'(X)+f'(V'\setminus X)\geq f'(\{x^*\})+f'(V_{n-1}\setminus\{v_1\}),$$
where $x^*:=\arg\min\{f'(\{x\})\mid x\in X\}$. This means that 
\begin{equation}\label{eq:ind}
f(X\cup\{v_1\})+f(V\setminus X)\geq f(\{v_1,x^*\})+f(V_{n-1}).    
\end{equation}
In addition, by the submodularity of $f$, we have $f(X\cup\{v_1\})-f(X)\leq f(\{v_1,x^*\})-f(\{x^*\})$, which together with \eqref{eq:ind} implies that 
$$f(X)+f(V\setminus X)\geq f(\{x^*\})+f(V_{n-1})\geq\min_{x\in X}f(\{x\})+f(V_{n-1}).$$
Otherwise, since $v_1\notin V\setminus X$, we have  
$$f(X)+f(V\setminus X)\geq\min_{z\in V\setminus X}f(\{z\})+f(V_{n-1})
\geq f(\{v_1\})+f(V_{n-1})\geq\min_{x\in X}f(\{x\})+f(V_{n-1}),$$
where the second inequality follows from the choice of $v_1$ and the third from $v_1\in X$. Thus, in either case, \eqref{eq:cut-0-order} holds. 
\end{proof}

\begin{theorem}
\label{th:MC-ordering}
Let $f:2^V\to\R$ be a symmetric submodular function and $(v_1,\ldots,v_n)$ be an MC-ordering of $V$ with respect to $f$. Then $(v_{n-1},v_n)$ is a contractible pair. 
\end{theorem}
\begin{proof}
Since $f$ is a symmetric submodular function, Lemma~\ref{lem:MC-ordering} implies that 
\begin{equation*}
2f(X)=f(X)+f(V\setminus X)\geq \min_{x\in X} f(\{x\})+f(\{v_n\})\geq 2\min_{x\in V} f(\{x\}) 
\end{equation*}
holds for any subset $X\subsetneq V$ with $|X\cap\{v_{n-1},v_n\}|=1$. 
Thus, $(v_{n-1},v_n)$ is a contractible pair. 
\end{proof}

\section{A Parametrized Family of Orderings}
\label{sec:alpha-ordering}
In this section, we introduce a family of orderings parametrized by $\alpha\in\R$, and show that previously known results on MA-ordering and MD-ordering can be understood as special cases of a unified framework. 

\begin{definition}[$\alpha$-Ordering]
  An ordering $(v_1, \dots, v_n)$ of $V$ is called an $\alpha$-ordering with respect to a set function $f$ if
  \begin{equation}\label{def:alpha-order}
    f(V_{i - 1} \cup \{v_i\}) + \alpha f(\{v_i\})
    \leq f(V_{i - 1} \cup \{v_j\}) + \alpha f(\{v_j\}) 
  \end{equation}
  holds for every pair $(i,j)$ such that $1 \leq i \leq j \leq n$. 
\end{definition}

For $\alpha = -1$, $\alpha = 0$, and $\alpha=1$, the definition of $\alpha$-ordering is tantamount to that of MA-ordering, MC-ordering, and MD-ordering, respectively.

For a submodular function $f:2^V\to\R$, let $g:2^V\to\R$ be a set function defined by 
\begin{equation*}\label{eq:g}
g(S):=f(S)+\alpha\sum_{v\in S}f(\{v\})\quad\quad(\forall S\subseteq V).
\end{equation*}
Obviously, $g$ is also a submodular function. In addition, we have the following lemma. 
\begin{lemma}\label{lem:0-to-alpha}
  Any $\alpha$-ordering $(v_1, \dots, v_n)$ with respect to $f$ 
  is an MC-ordering with respect to $g$. 
\end{lemma}

\begin{proof}
For any $1 \leq i \leq j \leq n$, we have
\begin{eqnarray*}
g(V_{i - 1} \cup \{v_i\})
    & = & f(V_{i - 1} \cup \{v_i\}) + \alpha f(\{v_i\}) + \alpha \sum_{h=1}^{i-1} f(\{v_h\}) \\
    & \leq & f(V_{i - 1} \cup \{v_j\}) + \alpha f(\{v_j\}) + \alpha \sum_{h=1}^{i-1} f(\{v_h\}) \\
    & = &  g(V_{i - 1} \cup \{v_j\}),
\end{eqnarray*}
where the inequality follows from \eqref{def:alpha-order}. 
Thus, $(v_1, \dots, v_n)$ is an MC-ordering of $V$ with respect to $g$.
\end{proof}

\begin{theorem}\label{thm:alpha-order-cut}
Let $f: 2^V \to \mathbb{R}$ be a symmetric submodular function and $(v_1, \dots, v_n)$ be an $\alpha$-ordering of $V$. Then for an arbitrary subset 
$X \subsetneq V$ with $|X \cap \{v_{n - 1}, v_n\}| =1$, we have  \begin{equation}\label{eq:alpha-order-cut_0}
  2f(X) \geq \min_{x \in X} (1 + \alpha)f(\{x\}) + (1 - \alpha)f(\{v_n\}).
\end{equation} 
\end{theorem}

\begin{proof}
Lemma~\ref{lem:0-to-alpha} implies that $(v_1, \dots, v_n)$ is an MC-ordering of $V$ with respect to $g$. Since $g$ is a submodular function, it follows 
from Lemma~\ref{lem:MC-ordering} that 
\begin{equation}\label{eq:alpha-order-cut_1}
    g(X) + g(V \setminus X) \geq \min_{x \in X} g(\{x\}) + g(V_{n-1})
\end{equation}
holds. By the definition of $g$, we have 
\begin{align}
   \label{eq:alpha-order-cut_2}
    g(X) + g(V \setminus X)
    &= f(X) + f(V \setminus X) + \alpha \sum_{v \in V} f(\{v\}), \\
    \label{eq:alpha-order-cut_3}
    \min_{x \in X} g(\{x\}) + g(V_{n - 1}) 
    &= \min_{x \in X}(1 + \alpha)f(\{x\}) + f(V_{n - 1}) + \alpha\sum_{h=1}^{n-1} f(\{v_h\}).
\end{align}
Combining~\eqref{eq:alpha-order-cut_1}--\eqref{eq:alpha-order-cut_3} and the  symmetry of $f$, we have
\begin{eqnarray*}
2f(X) & = & f(X) + f(V \setminus X) \\ 
    & \geq & \min_{x \in X}(1 + \alpha)f(\{x\})
    + f(V_{n - 1}) + \alpha\sum_{h=1}^{n-1} f(\{v_h\})
    - \alpha \sum_{v \in V} f(\{v\}) \\
    & = & \min_{x \in X}(1 + \alpha)f(\{x\}) + (1 - \alpha)f(\{v_n\}),
\end{eqnarray*}
which proves~\eqref{eq:alpha-order-cut_0}.
\end{proof}

Note that Theorem~\ref{thm:alpha-order-cut} unifies Theorems \ref{thm:ma-order-cut} and \ref{thm:md-order-cut}. In fact, the case $\alpha = -1$ recovers the pendent-pair result, while the case $\alpha = 1$ recovers the flat-pair result.

\begin{corollary}\label{cor:find-contractible}
  Let $f: 2^V \to \mathbb{R}$ be a symmetric submodular function
  and $(v_1, \dots, v_n)$ be an $\alpha$-ordering of $V$ with respect to $f$.
  If $\alpha \in [-1, 1]$, then $(v_{n - 1}, v_n)$ is a contractible pair.
\end{corollary}

\begin{proof}
Since $\alpha\in [-1, 1]$, we have $1 + \alpha \geq 0$ and $1 - \alpha \geq 0$.
By Theorem~\ref{thm:alpha-order-cut}, for any subset $X \subsetneq V$ with $|X \cap \{v_{n - 1}, v_n\}| = 1$, we have  
\begin{eqnarray*}
2f(X) & \geq & \min_{x \in X}(1 + \alpha)f(\{x\}) + (1 - \alpha)f(\{v_n\}) \\
      & \geq & \min_{v \in V}(1 + \alpha)f(\{v\}) + (1 - \alpha)\min_{v \in V} f(\{v\}) \\
      & = & 2 \min_{v \in V} f(\{v\}).
\end{eqnarray*}
Thus, $(v_{n - 1}, v_n)$ is a contractible pair. 
\end{proof}

\section{Conditions for Obtaining Special Pairs from $\alpha$-Ordering}
\label{sec:cond}
In this section, we examine the ranges of $\alpha$ that ensure $\alpha$-ordering to yield contractible, pendent, and flat pairs. 

As shown in Corollary~\ref{cor:find-contractible}, if $\alpha\in[-1,1]$, then the last two elements of an $\alpha$-ordering form a contractible pair. The following proposition shows the converse.  

\begin{proposition}\label{prop:requirement-for-contractible}
  For any $\alpha \notin [-1, 1]$, there exists a symmetric submodular function
  $f : 2^V \to \mathbb{R}$ and an $\alpha$-ordering $(v_1, \dots, v_n)$
  of $V$ with respect to $f$ such that $(v_{n - 1}, v_n)$ is not a contractible pair.
\end{proposition}

\begin{proof}
  We construct a symmetric submodular function $f$ 
  as the cut function of an undirected graph $G = (V, E)$ with weight function $c : E \to \R_+$.
  We distinguish three cases. 
  
  \smallskip
  \noindent\textbf{Case 1 ($\alpha > 1$).}
  Let $G=(V,E)$ be a graph with the vertex set $V = \{p_1, p_2, p_3, p_4\}$ and the edge set $E = \left\{\{p_1, p_4\}, \{p_2, p_3\}, \{p_3, p_4\}\right\}$.
  Define the edge capacity function $c : E \to \mathbb{R}_+$ by
  \begin{equation*}
    c(p_1, p_4) = c(p_2, p_3) = 1,\quad
    c(p_3, p_4) = \kappa_\alpha,
  \end{equation*}
  where $\kappa_\alpha = \frac{2}{1 + \alpha}$. Since $\alpha > 1$, we have $0 < \kappa_\alpha < 1$. \par
  We show that there exists an $\alpha$-ordering $(v_1, v_2, v_3, v_4)$ of $V$
  with respect to $f$ such that $\{v_3, v_4\} = \{p_3, p_4\}$.
  First, for each $v \in V$, we have
  \begin{equation*}
    f(\{p_1\}) = 1,\quad
    f(\{p_2\}) = 1,\quad
    f(\{p_3\}) = 1 + \kappa_\alpha,\quad
    f(\{p_4\}) = 1 + \kappa_\alpha.
  \end{equation*}
  Since $\alpha > -1$ and $\kappa_\alpha > 0$, we observe $\arg\min_{v\in V}(1 + \alpha)f(\{v\}) = \{p_1, p_2\}$. 
  Thus, we can assume that $v_1 = p_1$.
  Next, for each $v \in V \setminus \{p_1\}$, we have
  \begin{align*}
    f(\{p_1, p_2\}) + \alpha f(\{p_2\}) &= 2 + \alpha, \\
    f(\{p_1, p_3\}) + \alpha f(\{p_3\}) &= 2 + \kappa_\alpha + \alpha(1 + \kappa_\alpha), \\
    f(\{p_1, p_4\}) + \alpha f(\{p_4\}) &= \kappa_\alpha + \alpha(1 + \kappa_\alpha).
  \end{align*}
  Since $\kappa_\alpha=\frac{2}{1 + \alpha} > 0$, we obtain
  \begin{equation*}
    2 + \alpha
    = \frac{2(1 + \alpha)}{1 + \alpha} + \alpha
    = \kappa_\alpha + \alpha(1 + \kappa_\alpha)
    < 2 + \kappa_\alpha + \alpha(1 + \kappa_\alpha).
  \end{equation*}
  Therefore, $\arg\min_{v \in V \setminus \{p_1\}} \left(f(\{p_1, v\}) + \alpha f(\{v\})\right) = \{p_2, p_4\}$.
  Thus, we can assume that $v_2 = p_2$.
  It follows that $\{v_3, v_4\} = \{p_3, p_4\}$. \par
  We now consider $X = \{p_1, p_4\}$, which satisfies $|X \cap \{p_3, p_4\}| = 1$.
  Moreover, since $\kappa_\alpha < 1$,
  \begin{equation*}
    f(X) = \kappa_\alpha < 1 = f(\{p_1\}) = \min_{v \in V} f(\{v\}).
  \end{equation*}
  Therefore, $(p_3, p_4)$ is not a contractible pair.

  \smallskip
  \noindent\textbf{Case 2 ($-2 < \alpha < -1$).}
  Let $G=(V,E)$ be a graph with the vertex set $V = \{p_1, p_2, p_3, p_4, p_5\}$ and the edge set 
  $E = \left\{\{p_1, p_2\}, \{p_1, p_3\}, \{p_1, p_5\}, \{p_3, p_4\}, \{p_3, p_5\}\right\}$.
  Define the edge capacity function $c : E \to \mathbb{R}_+$ by
  \begin{equation*}
    c(p_1, p_3) = c(p_3, p_5) = 1,\quad
    c(p_1, p_5) = 2,\quad
    c(p_1, p_2) = c(p_3, p_4) = \kappa_\alpha,
  \end{equation*}
  where $\kappa_\alpha = 1 + \frac{2}{-1 - \alpha}$. Since $-2 < \alpha < -1$, we have $\kappa_\alpha > 3$. \par
  We show that there exists an $\alpha$-ordering $(v_1, v_2, v_3, v_4, v_5)$ of $V$
  with respect to $f$ such that $\{v_4, v_5\} = \{p_4, p_5\}$.
  First, for each $v \in V$, we have
  \begin{equation*}
    f(\{p_1\}) = 3 + \kappa_\alpha,\quad
    f(\{p_2\}) = \kappa_\alpha,\quad
    f(\{p_3\}) = 2 + \kappa_\alpha,\quad
    f(\{p_4\}) = \kappa_\alpha,\quad
    f(\{p_5\}) = 3.
  \end{equation*}
  Since $\alpha < -1$ and $\kappa_\alpha > 3$, we obtain $\arg\min_{v\in V}(1 + \alpha)f(\{v\}) = \{p_1\}$, which implies $v_1 = p_1$.
  Next, for each $v \in V \setminus \{p_1\}$, we have
  \begin{align*}
    f(\{p_1, p_2\}) + \alpha f(\{p_2\}) &= 3 + \alpha \kappa_\alpha, \\
    f(\{p_1, p_3\}) + \alpha f(\{p_3\}) &= 3 + 2\kappa_\alpha + \alpha(2 + \kappa_\alpha), \\
    f(\{p_1, p_4\}) + \alpha f(\{p_4\}) &= 3 + 2\kappa_\alpha + \alpha \kappa_\alpha, \\
    f(\{p_1, p_5\}) + \alpha f(\{p_5\}) &= 2 + \kappa_\alpha + 3\alpha.
  \end{align*}
  Since $-2 < \alpha < -1$ and $\kappa_\alpha > 3$, we obtain
  \begin{align*}
    \left(3 + 2\kappa_\alpha + \alpha(2 + \kappa_\alpha)\right) - \left(3 + \alpha \kappa_\alpha\right)
    &= 2(\alpha + \kappa_\alpha)
    > 2(-2 + 3)
    > 0, \\
    \left(3 + 2\kappa_\alpha + \alpha \kappa_\alpha\right) - \left(3 + \alpha \kappa_\alpha\right)
    &= 2\kappa_\alpha
    > 0, \\
    \left(2 + \kappa_\alpha + 3\alpha\right) - \left(3 + \alpha \kappa_\alpha\right)
    &= (1 - \alpha)\kappa_\alpha + 3\alpha - 1 
    = \frac{2\left(2 - (\alpha + 1)^2\right)}{-\alpha - 1}
    > 0.
  \end{align*}
  Therefore, $\arg\min_{v \in V \setminus \{p_1\}} \left(f(\{p_1, v\}) + \alpha f(\{v\})\right) = \{p_2\}$.
  Thus, we have $v_2 = p_2$.
  Moreover, for each $v \in V \setminus \{p_1, p_2\}$, we have
  \begin{align*}
    f(\{p_1, p_2, p_3\}) + \alpha f(\{p_3\}) &= 3 + \kappa_\alpha + \alpha(2 + \kappa_\alpha), \\
    f(\{p_1, p_2, p_4\}) + \alpha f(\{p_4\}) &= 3 + \kappa_\alpha + \alpha \kappa_\alpha, \\
    f(\{p_1, p_2, p_5\}) + \alpha f(\{p_5\}) &= 2 + 3\alpha.
  \end{align*}
  Since $-2 < \alpha < -1$ and $\kappa_\alpha > 3$, we obtain
  \begin{align*}
    \left(3 + \kappa_\alpha + \alpha \kappa_\alpha\right) - \left(3 + \kappa_\alpha + \alpha(2 + \kappa_\alpha)\right)
    &= -2\alpha
    > 0, \\
    \left(2 + 3\alpha\right) - \left(3 + \kappa_\alpha + \alpha(2 + \kappa_\alpha)\right)
    &= (1 + \alpha)\kappa_\alpha + (1 - \alpha)
    = 0.
  \end{align*}
  Therefore, we have $\arg\min_{v \in V \setminus \{p_1, p_2\}} \left(f(\{p_1, p_2, v\}) + \alpha f(\{v\})\right) = \{p_3, p_5\}$.
  Thus, we can assume that $v_3 = p_3$.
  It follows that $\{v_4, v_5\} = \{p_4, p_5\}$. \par
  We now consider $X = \{p_3, p_4\}$, which satisfies $|X \cap \{p_4, p_5\}| = 1$.
  Moreover, since $\kappa_\alpha > 3$,
  \begin{equation*}
    f(X) = 2 < 3 = f(\{p_5\}) = \min_{v \in V} f(\{v\}).
  \end{equation*}
  Therefore, $(p_4, p_5)$ is not a contractible pair.

  \smallskip
  \noindent\textbf{Case 3 ($\alpha < -2$).}
  Let $G=(V,E)$ be a graph with the vertex set $V = \{p_1, p_2, p_3, p_4, p_5\}$ and the edge set $E = \left\{\{p_1, p_2\}, \{p_1, p_3\}, \{p_1, p_5\}, \{p_3, p_4\}, \{p_3, p_5\}\right\}$.
  Define the edge capacity function $c : E \to \mathbb{R}_+$ by
  \begin{equation*}
    c(p_1, p_3) = c(p_3, p_5) = 1,\quad
    c(p_1, p_5) = 2,\quad
    c(p_1, p_2) = 7,\quad
    c(p_3, p_4) = 4.
  \end{equation*}
  \par
  We show that there exists an $\alpha$-ordering $(v_1, v_2, v_3, v_4, v_5)$ of $V$
  with respect to $f$ such that $\{v_4, v_5\} = \{p_4, p_5\}$.
  First, for each $v \in V$, we have
  \begin{equation*}
    f(\{p_1\}) = 10,\quad
    f(\{p_2\}) = 7,\quad
    f(\{p_3\}) = 6,\quad
    f(\{p_4\}) = 4,\quad
    f(\{p_5\}) = 3.
  \end{equation*}
  Since $\alpha < -2$, we obtain $\arg\min_{v\in V}(1 + \alpha)f(\{v\}) = \{p_1\}$, which implies $v_1 = p_1$.
  Next, for each $v \in V \setminus \{p_1\}$, we have
  \begin{align*}
    f(\{p_1, p_2\}) + \alpha f(\{p_2\}) &= 3 + 7\alpha, \\
    f(\{p_1, p_3\}) + \alpha f(\{p_3\}) &= 14 + 6\alpha, \\
    f(\{p_1, p_4\}) + \alpha f(\{p_4\}) &= 14 + 4\alpha, \\
    f(\{p_1, p_5\}) + \alpha f(\{p_5\}) &= 9 + 3\alpha.
  \end{align*}
  Since $\alpha < 0$, we obtain
  \begin{align*}
    \left(14 + 6\alpha\right) - \left(3 + 7\alpha\right)
    &= 11 - \alpha
    > 0, \\
    \left(14 + 4\alpha\right) - \left(3 + 7\alpha\right)
    &= 11 - 3\alpha
    > 0, \\
    \left(9 + 3\alpha\right) - \left(3 + 7\alpha\right)
    &= 6 - 4\alpha
    > 0. \\
  \end{align*}
  Therefore, $\arg\min_{v \in V \setminus \{p_1\}} \left(f(\{p_1, v\}) + \alpha f(\{v\})\right) = \{p_2\}$, which implies $v_2 = p_2$.
  Moreover, for each $v \in V \setminus \{p_1, p_2\}$, we have
  \begin{align*}
    f(\{p_1, p_2, p_3\}) + \alpha f(\{p_3\}) &= 7 + 6\alpha, \\
    f(\{p_1, p_2, p_4\}) + \alpha f(\{p_4\}) &= 7 + 4\alpha, \\
    f(\{p_1, p_2, p_5\}) + \alpha f(\{p_5\}) &= 2 + 3\alpha.
  \end{align*}
  Since $\alpha < -2$, we obtain
  \begin{align*}
    \left(7 + 4\alpha\right) - \left(7 + 6\alpha\right)
    &= -2\alpha
    > 0, \\
    \left(2 + 3\alpha\right) - \left(7 + 6\alpha\right)
    &= -5 - 3\alpha
    > 0. 
  \end{align*}
  Therefore, we have $\arg\min_{v \in V \setminus \{p_1, p_2\}} \left(f(\{p_1, p_2, v\}) + \alpha f(\{v\})\right) = \{p_3\}$, which implies $v_3 = p_3$.
  It follows that $\{v_4, v_5\} = \{p_4, p_5\}$. \par
  We now consider $X = \{p_3, p_4\}$, which satisfies $|X \cap \{p_4, p_5\}| = 1$. On the other hand,
  \begin{equation*}
    f(X) = 2 < 3 = f(\{p_5\}) = \min_{v \in V} f(\{v\}).
  \end{equation*}
  Therefore, $(p_4, p_5)$ is not a contractible pair.
\end{proof}

We now investigate the ranges of $\alpha$ for pendent and flat pairs. 

\begin{proposition}\label{prop:requirement-for-pendent}
  For any $\alpha \neq -1$, there exists a symmetric submodular function
  $f : 2^V \to \mathbb{R}$ and an $\alpha$-ordering $(v_1, \dots, v_n)$
  of $V$ with respect to $f$ such that $(v_{n - 1}, v_n)$ is not a pendent pair.
\end{proposition}

\begin{proof}
  Since every pendent pair is contractible, Proposition~\ref{prop:requirement-for-contractible} implies that for any $\alpha \notin [-1, 1]$, there exists
  a symmetric submodular function $f : 2^V \to \mathbb{R}$ and an
  $\alpha$-ordering $(v_1, \dots, v_n)$ of $V$ such that $(v_{n - 1}, v_n)$ is not
  a pendent pair. Therefore, it suffices to consider the case $\alpha \in (-1, 1]$.
  Consider $V = \{p_1, p_2, p_3, p_4\}$ and $E = \left\{\{p_1, p_4\}, \{p_2, p_3\}, \{p_3, p_4\}\right\}$, and 
  define the edge capacity function $c: E \to \R_+$ by
  \begin{equation*}
    c(p_1, p_4) = c(p_2, p_3) = 1,\quad
    c(p_3, p_4) = \kappa_\alpha,
  \end{equation*}
  where $\kappa_\alpha = \frac{2}{1 + \alpha}$. Since $\alpha \in (-1, 1]$, we have $\kappa_\alpha > 0$. \par
  By the same argument as in Case~1 of the proof of
  Proposition~\ref{prop:requirement-for-contractible}, there exists an
  $\alpha$-ordering $(v_1, v_2, v_3, v_4)$ of $V$ with respect to $f$
  such that $v_1 = p_1$, $v_2 = p_2$, and hence $\{v_3, v_4\} = \{p_3, p_4\}$.
  We now consider $X = \{p_1, p_4\}$, which satisfies $|X \cap \{p_3, p_4\}| = 1$.
  Moreover, since $\kappa_\alpha > 0$, we have 
  \begin{equation*}
    f(X) = \kappa_\alpha < 1 + \kappa_\alpha = f(\{p_3\}) = f(\{p_4\}).
  \end{equation*}
  Hence, regardless of whether $v_4 = p_3$ or $v_4 = p_4$, we have $f(X) < f(\{v_4\})$.
  Therefore, $(p_3, p_4)$ is not a pendent pair.
\end{proof}

\begin{proposition}\label{prop:requirement-for-flat}
  For any $\alpha \neq 1$, there exists a symmetric submodular function
  $f : 2^V \to \mathbb{R}$ and an $\alpha$-ordering $(v_1, \dots, v_n)$
  of $V$ with respect to $f$ such that $(v_{n - 1}, v_n)$ is not a flat pair.
\end{proposition}

\begin{proof}
  Since every flat pair is contractible, Proposition \ref{prop:requirement-for-contractible} implies that for any $\alpha \notin [-1, 1]$, there exists
  a symmetric submodular function $f : 2^V \to \mathbb{R}$ and an
  $\alpha$-ordering $(v_1, \dots, v_n)$ of $V$ such that $(v_{n - 1}, v_n)$ is not
  a flat pair. Therefore, it suffices to consider the case $\alpha \in [-1, 1)$.
  Consider $V = \{p_1, p_2, p_3, p_4, p_5\}$ and 
  $E = \left\{\{p_1, p_2\}, \{p_2, p_3\}, \{p_3, p_4\}, \{p_4, p_5\}\right\}$, and define the edge capacity function $c : E \to \mathbb{R}_+$ by
  \begin{equation*}
    c(p_1, p_2) = c(p_2, p_3) = c(p_4, p_5) = 1,\quad
    c(p_3, p_4) = \kappa_\alpha,
  \end{equation*}
  where
  \begin{equation*}
    \kappa_\alpha
    \in \left(1, \frac{2}{1 + \alpha}\right]
    \quad \text{if } \alpha \in (-1, 1),
  \end{equation*}
  and $\kappa_\alpha > 1$ if $\alpha = -1$.\par
  We show that there exists an $\alpha$-ordering $(v_1, v_2, v_3, v_4, v_5)$ of $V$
  with respect to $f$ such that $\{v_4, v_5\} = \{p_4, p_5\}$.
  First, for each $v \in V$, we have
  \begin{equation*}
    f(\{p_1\}) = 1,\quad
    f(\{p_2\}) = 2,\quad
    f(\{p_3\}) = 1 + \kappa_\alpha,\quad
    f(\{p_4\}) = 1 + \kappa_\alpha,\quad
    f(\{p_5\}) = 1.
  \end{equation*}
  If $\alpha \in (-1, 1)$, then since $\kappa_\alpha > 1$, we have 
  $\arg\min_{v\in V}(1 + \alpha)f(\{v\}) = \{p_1, p_5\}$. 
  Thus, we can assume that $v_1 = p_1$.
  If $\alpha = -1$, then since $1 + \alpha = 0$, so we can also assume that $v_1 = p_1$.
  Next, for each $v \in V \setminus \{p_1\}$, we have
  \begin{align*}
    f(\{p_1, p_2\}) + \alpha f(\{p_2\}) &= 1 + 2\alpha, \\
    f(\{p_1, p_3\}) + \alpha f(\{p_3\}) &= 2 + \kappa_\alpha + \alpha(1 + \kappa_\alpha), \\
    f(\{p_1, p_4\}) + \alpha f(\{p_4\}) &= 2 + \kappa_\alpha + \alpha(1 + \kappa_\alpha), \\
    f(\{p_1, p_5\}) + \alpha f(\{p_5\}) &= 2 + \alpha.
  \end{align*}
  Since $\alpha \in [-1, 1)$ and $\kappa_\alpha > 1$, we obtain
  \begin{align*}
    \left(2 + \kappa_\alpha + \alpha(1 + \kappa_\alpha)\right) - \left(1 + 2\alpha\right)
    &= 1 - \alpha + (1 + \alpha)\kappa_\alpha
    > 0, \\
    \left(2 + \alpha\right) - \left(1 + 2\alpha\right)
    &= 1 - \alpha
    > 0. 
  \end{align*}
  Therefore, we have $\arg\min_{v \in V \setminus \{p_1\}} \left(f(\{p_1, v\}) + \alpha f(\{v\})\right) = \{p_2\}$, and hence $v_2 = p_2$.
  Moreover, for each $v \in V \setminus \{p_1, p_2\}$, we have
  \begin{align*}
    f(\{p_1, p_2, p_3\}) + \alpha f(\{p_3\}) &= \kappa_\alpha + \alpha(1 + \kappa_\alpha), \\
    f(\{p_1, p_2, p_4\}) + \alpha f(\{p_4\}) &= 2 + \kappa_\alpha + \alpha(1 + \kappa_\alpha), \\
    f(\{p_1, p_2, p_5\}) + \alpha f(\{p_5\}) &= 2 + \alpha.
  \end{align*}
  Since $\alpha \in [-1, 1)$ and $\kappa_\alpha > 1$, we have
  \begin{equation*}
    \left(2 + \kappa_\alpha + \alpha(1 + \kappa_\alpha)\right) - \left(\kappa_\alpha + \alpha(1 + \kappa_\alpha)\right)
    = 2
    > 0.
  \end{equation*}
  If $\alpha \in (-1, 1)$, then since $\kappa_\alpha \leq \frac{2}{1 + \alpha}$, we obtain
  \begin{equation*}
    \left(2 + \alpha\right) - \left(\kappa_\alpha + \alpha(1 + \kappa_\alpha)\right)
    = 2 - (1 + \alpha) \kappa_\alpha
    \geq 2 - \frac{2(1 + \alpha)}{1 + \alpha}
    = 0.
  \end{equation*}
  If $\alpha = -1$, then
  \begin{equation*}
    \left(2 + \alpha\right) - \left(\kappa_\alpha + \alpha(1 + \kappa_\alpha)\right)
    = 2 - (1 + \alpha) \kappa_\alpha 
    = 2
    > 0.
  \end{equation*}
  Therefore, we have $\arg\min_{v \in V \setminus \{p_1, p_2\}} \left(f(\{p_1, p_2, v\}) + \alpha f(\{v\})\right) = \{p_3\}$, and hence $v_3 = p_3$.
  It follows that $\{v_4, v_5\} = \{p_4, p_5\}$. \par
  We now consider $X = \{p_3, p_4\}$, which satisfies $|X \cap \{p_4, p_5\}| = 1$.
  Moreover, since $\kappa_\alpha > 1$, we have 
  \begin{equation*}
    f(X) = 2 < 1 + \kappa_\alpha = f(\{p_3\}) = \min_{x \in X} f(\{x\}).
  \end{equation*}
  Therefore, $(p_4, p_5)$ is not a flat pair.
\end{proof}

\bibliographystyle{plain}
\bibliography{references}

\end{document}

%% file: _config.tex
\usepackage[margin=1in]{geometry}
\usepackage{amsmath,amssymb,amsthm}
\usepackage[hidelinks]{hyperref}
\usepackage{theoremref}
\usepackage{mathtools}
\usepackage{algorithm}
\usepackage{algpseudocode}
\usepackage{subcaption}
\usepackage{tikz}
\usepackage{multirow}
\usepackage{booktabs}
\usepackage{adjustbox}
\usepackage{physics}
\usepackage{float}
\usepackage{enumitem}

\newtheorem{theorem}{Theorem}
\newtheorem{lemma}{Lemma}

\newtheorem{corollary}{Corollary}
\newtheorem{proposition}{Proposition}
\newtheorem{definition}{Definition}

\usetikzlibrary{arrows}
\usetikzlibrary{shapes.geometric}
\usetikzlibrary{positioning}
\usetikzlibrary{intersections, calc, arrows}

\makeatletter
\newcommand{\figcaption}[1]{\def\@captype{figure}\caption{#1}}
\newcommand{\tblcaption}[1]{\def\@captype{table}\caption{#1}}
\makeatother

\let\oldComment\Comment
\renewcommand{\Comment}[1]{\oldComment{\textcolor{gray}{#1}}}